\documentclass[journal]{IEEEtran}

\ifCLASSINFOpdf
\else
   \usepackage[dvips]{graphicx}
\fi
\usepackage{url}

\usepackage{geometry}
\usepackage{color} 
\usepackage{subfigure} 
\usepackage{amsmath,lipsum}

\usepackage{stfloats}

\usepackage{bm}

\usepackage{graphicx}

\usepackage{float} 

\usepackage{amssymb}

\newtheorem{theorem}{Theorem}

\newtheorem{lemma}{Lemma}

\newtheorem{proposition}{Proposition}

\newtheorem{corollary}{Corollary}

\newtheorem{remark}{Remark}

\usepackage[noend]{algpseudocode}

\usepackage{algorithmicx,algorithm}

\newenvironment{proof}{{\indent \indent \it Proof:}}{\hfill $\blacksquare$\par}

\begin{document}
\title{\huge Secure Communication for Spatially Correlated RIS-Aided Multiuser Massive MIMO Systems: Analysis and Optimization}

\author{
Dan Yang, Jindan Xu, Wei Xu, \IEEEmembership{Senior Member, IEEE,} Yongming Huang, \IEEEmembership{Senior Member, IEEE,} and Zhaohua Lu
\thanks{D. Yang, W. Xu, and Y. Huang are with the National Mobile Communications Research Laboratory, Southeast University, Nanjing 210096, China (email: {dyang@seu.edu.cn; wxu@seu.edu.cn; huangym@seu.edu.cn}). 

J. Xu is with the Engineering Product Development Pillar, Singapore University of Technology and Design, Singapore 487372 (e-mail: jindan xu@sutd.edu.sg).

Z. Lu is with ZTE Corporation, and State Key Laboratory of Mobile Network and Mobile Multimedia Technology, Shenzhen, China (email: {lu.zhaohua@zte.com.cn})
}
}
\maketitle

\begin{abstract}
This letter investigates the secure communication in a reconfigurable intelligent surface (RIS)-aided multiuser massive multiple-input multiple-output (MIMO) system exploiting artificial noise (AN). We first derive a closed-form expression of the ergodic secrecy rate under spatially correlated MIMO channels. By using this derived result, we further optimize the power fraction of AN in closed form and the RIS phase shifts by developing a gradient-based algorithm, which requires only statistical channel state information (CSI). Our analysis shows that spatial correlation at the RIS provides an additional dimension for optimizing the RIS phase shifts. Numerical simulations validate the analytical results which show the insightful interplay among the system parameters and the degradation of secrecy performance due to high spatial correlation at the RIS.
\end{abstract}
%As a result, we achieve a reduced complexity and overhead for joint optimization instead of frequent computation in conventional MIMO systems with instantaneous CSI. 

\setlength{\parskip}{0\baselineskip}

\section{Introduction}
\IEEEPARstart{R}{econfigurable} intelligent surface (RIS) has been proposed as a promising technology for improving both spectral and energy efficiencies for next-generation wireless networks [1]. It consists of numerous low-cost passive reflecting elements that can induce phase changes to the reflected electromagnetic waves. As such, by properly adjusting the phase shifts, RIS can smartly modify the channel conditions between the base station (BS) and the users, which helps improve the communication quality of wireless networks [2].

%. It is known that RIS can improve the secrecy rate by increasing the signal power of legitimate users while decreasing the signal power of the eavesdropper (Eve)
Recently, there has been considerable interest in the use of RIS to enhance the physical layer security of wireless communication networks [3]--[6]. In [3], the transmit beamforming jointly with artificial noise (AN) and the RIS phase shifts was optimized under a multiple-input multiple-output (MIMO) wiretap channel. It was shown that the secrecy rate performance can be strengthened with the aid of AN in RIS-assisted systems. As for multiuser scenarios, the authors of [4] investigated the robust secrecy design by solving a transmit power minimization problem. Besides, in [5], the weighted sum secrecy rate was maximized by taking both the direct link and the cascaded RIS link into account. In [6], a RIS-aided secure multiple-input single-output (MISO) communication system was studied, where multiple colluding Eves coexist. 

However, the design of RIS phase shifts of most methods was based on instantaneous channel state information (CSI) in each interval, which can be unaffordable due to frequent phase adjustment at RIS and channel estimation at the BS. Secondly, all the above works were based on independent Rayleigh or Rician fading and have not considered the impact of spatial correlation on the secrecy performance. In fact, spatial correlation generally exists at the RIS due to physical constraints in a rectangular array, and it varies by adjusting the space among adjacent RIS elements [7], [8]. Moreover, previous works on RIS-aided systems, e.g., [5], [6], are usually restricted to single-antenna Eves scenarios for the sake of analytical tractability.

Against the above background, the performance of spatially correlated RIS-aided multiuser massive MIMO systems is first studied in the presence of a multi-antenna Eve. The main contributions of this work are listed below.
\begin{itemize}
	\item We derive the closed-form expression for the ergodic secrecy rate, which depends only on statistical channel information of the users and Eve.
	\item We devise an alternating algorithm to maximize the ergodic secrecy rate, where the power fraction of AN is optimally obtained in closed form and the RIS phase shifts are designed by a projected gradient ascent method.
	\item Insightful observations of the impact of spatial correlation and the number of RIS elements on the secrecy performance are presented. It indicates that the spatial correlation enhances the ability of RIS to properly manipulate the wireless environment.
\end{itemize}

\emph{Notation}: The inverse, conjugate transpose, and trace of matrix $\bf A$ are denoted by ${\bf A}^{-1}$, ${\bf A}^H$, ${\rm tr}({\bf A})$, respectively. $\mathcal {CN}(0,\sigma^2)$ represents the complex Gaussian distribution with zero mean and variance $\sigma^2$. Besides, $\mathbb E\{\cdot\}$ and $\rm var\{\cdot\}$ denote the expectation and variance of a random variable, respectively. We use $\mathbb C^{m\times n}$ to represent the space of all $m\times n$ matrices with complex-valued elements. ${\bf I}_K$ denotes the $K$-dimensional identity matrix.

\section{System Model}
We consider a RIS-aided multiuser massive MIMO secure system, where $K$ single-antenna legitimate users are served by an $M$-antenna BS leveraging a RIS of $N$ reflecting elements. One passive Eve equipped with $M_{\rm E}$ antennas is located around users seeking to wiretap the transmitted information{\footnote{Note that this assumption also applies to situations where there are several Eves collusively eavesdropping on the same secrecy data [9].}}.
Assume that the RIS is controlled by the BS through a perfect backhaul link and perfect CSI of the users are available as it plays the role of an upper bound with imperfect CSI in practice. To evaluate the secrecy performance, channel distribution information of Eve is assumed to be available at the BS, which has been widely adopted and validated in literature, e.g., [9], [10].

We assume narrow-band quasi-static fading channels. Let ${\bf H}_1\in\mathbb C^{M\times N}$, ${\bf h}_{{\rm B},k}\in\mathbb C^{M\times 1}$, ${\bf H}_{\rm B,E}\in\mathbb C^{M\times M_{\rm E}}$, ${\bf h}_{{\rm I},k}\in\mathbb C^{N\times 1}$, and ${\bf H}_{\rm I,E}\in\mathbb C^{N\times M_{\rm E}}$, respectively, denote the channel from the BS to RIS, BS to user $k$, BS to Eve, RIS to user $k$, and RIS to Eve. Notably, we consider spatially correlated rather than independent Rayleigh fading. Hence, we have
{\abovedisplayskip=4pt plus 4pt minus 6pt
\abovedisplayshortskip=0pt plus 4pt
\belowdisplayskip=4pt plus 4pt minus 6pt
\belowdisplayshortskip=0pt plus 4pt minus 4pt
\begin{align}
	{\bf h}_{{\rm I},k}=\sqrt{\beta_{{\rm I},k}}\mathbf{R}_{{\rm I},k}^{1/2}{\bf g}_{{\rm I},k}\quad \ \ \ {\bf h}_{{\rm B},k}=\sqrt{\beta_{2,k}}\mathbf{R}_{{\rm B},k}^{1/2}{\bf g}_{{\rm B},k},\\
	{\bf H}_{\rm I,E}=\sqrt{\beta_{\rm I,E}}\mathbf{R}_{\rm I,E}^{1/2}\mathbf{G}_{\rm I,E}\quad \ {\bf H}_{\rm B,E}=\sqrt{\beta_3}\mathbf{R}_{\rm B,E}^{1/2}\mathbf{G}_{\rm B,E}
\end{align}}where $\beta_{{\rm I},k}$, $\beta_{\rm I,E}$, $\beta_3$ and $\beta_{2,k}$ represent the large-scale path losses of the corresponding channels. Elements of $\mathbf{g}_{{\rm I},k}$, $\mathbf{g}_{{\rm B},k}$, $\mathbf{G}_{\rm I,E}$, and $\mathbf{G}_{\rm B,E}$ are independently and identically distributed (i.i.d.) complex Gaussian random variables with zero mean and unit variance. In addition, $\mathbf{R}_{{\rm I},k}$ and $\mathbf{R}_{{\rm B},k}$ ($\mathbf{R}_{\rm I,E}$ and $\mathbf{R}_{\rm B,E}$) are respectively the channel correlation matrices at the RIS and BS. Moreover, the LoS channel ${\bf H}_1$ is modeled, similar to [11], as a full-rank channel matrix with $[{\bf H}_1]_{m,n}=\sqrt{\beta_1}e^{-j2\pi\frac{d_{m,n}}{\lambda}}$, where $\beta_1$ is the path loss, $\lambda$ is the carrier wavelength, and $d_{m,n}$ is the distance between reflecting element $m$ of the RIS and antenna $n$ of the BS. Such channels are popularly seen when deterministic scattering presents between the BS and RIS or placing the RIS close to the BS [11]. Note that the correlation matrices and the path losses are assumed to be known, e.g., by the methods in [12]. In addition, denote the phase shift matrix of the RIS by a diagonal matrix ${\bf \Phi}={\rm diag}(e^{j\theta_1},...,e^{j\theta_n},...,e^{j\theta_N})$, where $\theta_n\in[0, 2\pi)$ represents the phase shift of the $n$th RIS reflecting element\footnote{As usual, we use the amplitude-independent phase shift model for tractable analysis. The analysis based on practical amplitude models [13] will be left for our future work.}.

In case that the instantaneous CSI of Eve is completely unknown, AN is injected to mask the confidential information. Before transmission, the information signal ${\bf s}$ with $\mathbb E\{{\bf s}{\bf s}^H\}={\bf I}_K$ and the AN signal ${\bf z}\sim\mathcal{CN}({\bf 0}_{M-K},{\bf I}_{M-K})$ are multiplied by data precoder ${\bf W}\in\mathbb C^{M\times K}$ with ${\rm tr}({\bf W}{\bf W}^H)=K$ and AN precoder ${\bf V}\in\mathbb C^{M\times (M-K)}$ with ${\rm tr}({\bf V}{\bf V}^H)=M-K$, respectively. The transmit signal is expressed as
{\abovedisplayskip=4pt plus 2pt minus 6pt
\abovedisplayshortskip=0pt plus 2pt
\belowdisplayskip=4pt plus 2pt minus 6pt
\belowdisplayshortskip=2pt plus 2pt minus 4pt
\begin{align}
	{\bf x}=\sqrt{\frac{\xi P}{K}}{\bf Ws}+\sqrt{\frac{(1-\xi) P}{M-K}}{\bf Vz}\triangleq\sqrt{p}{\bf Ws}+\sqrt{q}{\bf Vz},
\end{align}}where $P$ denotes the total transmit power and $\xi\in[0,1]$ is the fraction of power allocated to the information ({\emph {power fraction}}, for short). Based on the above definitions, the transmit signal in (3) satisfies the power constraint $\mathbb E\{{\bf x}^H{\bf x}\}=P$. For notational simplicity, we define $p\triangleq\frac{\xi P}{K}$ and $q\triangleq\frac{(1-\xi) P}{M-K}$. Then, the received signals at user $k$ and Eve are respectively given by
{\abovedisplayskip=4pt plus 2pt minus 4pt
\abovedisplayshortskip=0pt plus 2pt
\belowdisplayskip=4pt plus 2pt minus 6pt
\belowdisplayshortskip=4pt plus 2pt minus 4pt
\begin{equation}
	y_k=\sqrt{p}{\bf h}_k^H{\bf w}_k s_k+\sqrt{p}\sum\nolimits_{i\neq k}{\bf h}_k^H{\bf w}_i s_i+\sqrt{q}{\bf h}_k^H{\bf Vz}+n_k,
\end{equation}
\begin{equation}
	{\bf y}_{\rm E}=\sqrt{p}{\bf H}_{\rm E}^H{\bf Ws}+\sqrt{q}{\bf H}_{\rm E}^H{\bf Vz}+{\bf n}_{\rm E},
\end{equation}}where $n_k\sim\mathcal{CN}(0,\sigma_k^2)$ and $\mathbf{n}_{\rm E}\sim\mathcal{CN}({\bf 0},\sigma_{\rm E}^2{\bf I}_{M_{\rm E}})$ are the additive white Gaussian noise (AWGN) at user $k$ and Eve, while ${\bf h}_k={\bf H}_1{\bf \Phi}{\bf h}_{{\rm I},k}+{\bf h}_{{\rm B},k}$ and ${\bf H}_{\rm E}={\bf H}_1{\bf \Phi}{\bf H}_{\rm I,E}+{\bf H}_{\rm B,E}$ represents the equivalent channel from the BS to user $k$ and to Eve, respectively. 
	
\section{Secrecy Performance Analysis}
In this section, the ergodic secrecy rate of the RIS-aided secure system is derived in closed form.

We take advantage of channel hardening, because users do not have any knowledge of the instantaneous CSI in practice, but they are aware of their statistics. Therefore, the received signal is decomposed as
{\abovedisplayskip=4pt plus 4pt minus 6pt
\abovedisplayshortskip=0pt plus 4pt
\belowdisplayskip=4pt plus 4pt minus 6pt
\belowdisplayshortskip=0pt plus 4pt minus 4pt
\begin{align}
y_k=&\mathbb E\{\sqrt{p}{\bf h}_k^H{\bf w}_k\} s_k+\left(\sqrt{p}{\bf h}_k^H{\bf w}_k-\mathbb E\{\sqrt{p}{\bf h}_k^H{\bf w}_k\}\right)s_k\nonumber\\
&+\sqrt{p}\sum\nolimits_{i\neq k}{\bf h}_k^H{\bf w}_i s_i+\sqrt{q}{\bf h}_k^H{\bf Vz}+n_k.
\end {align}}By treating the interference and channel uncertainty as an equivalent noise term, a lower bound for the achievable rate of user $k$ is given by
{\abovedisplayskip=2pt plus 2pt minus 6pt
\abovedisplayshortskip=0pt plus 2pt
\belowdisplayskip=2pt plus 2pt minus 6pt
\belowdisplayshortskip=0pt plus 2pt minus 4pt
\begin{equation}
R_k=\log_2\bigg(1+\frac{\left|\mathbb E\{\sqrt{p}{\bf h}_k^H{\bf w}_k\}\right|^2}{\Psi}\bigg),
\end{equation}}where $\Psi=\sum\nolimits_{i\neq k}\mathbb E\{|\sqrt{p}{\bf h}_k^H{\bf w}_i|^2\}+\mathbb E\{q{\bf h}_k^H{\bf V}{\bf V}^H{\bf h}_k\}+{\rm var}\{\sqrt{p}{\bf h}_k^H{\bf w}_k\}+\sigma_k^2$. For analytical tractability and low complexity, we adopt the MRT precoding [9], and ${\bf V}=[{\bf v}_1,...,{\bf v}_i,...,{\bf v}_{M-K}]$ with $\Vert{\bf v}_i\Vert=1, i=1,...,M-K$, is designed to lie in the null space of the user channels, i.e., ${\bf H}^H{\bf V} ={\bf0}$, where ${\bf H}=[{\bf h}_1,...,{\bf h}_K]$.
%i.e., ${\bf W}=\sqrt{\frac{K}{{{\rm tr}({\bf H}^H{\bf H})}}}{\bf H}$, where

Considering a pessimistic case, Eve is so powerful that it is perfectly aware of its channel and can remove all the interference from legitimate users, i.e., strongly eavesdropping in [10], [14]. Then, from (5), an upper bound for the capacity of Eve is obtained as
{\abovedisplayskip=4pt plus 2pt minus 6pt
\abovedisplayshortskip=0pt plus 2pt
\belowdisplayskip=4pt plus 2pt minus 6pt
\belowdisplayshortskip=2pt plus 2pt minus 4pt
\begin{equation}
	C=\mathbb E\big\{\log_2\big(1+p{\bf w}_k^H{\bf H}_{\rm E}{\bf X}^{-1}{\bf H}_{\rm E}^H{\bf w}_k\big)\big\},
\end{equation}}where ${\bf X}\triangleq q{\bf H}_{\rm E}^H{\bf V}{\bf V}^H{\bf H}_{\rm E}$ denotes the noise correlation matrix at Eve. In addition, since the noise level at Eve is unknown, it is reasonable to assume negligible thermal noise, i.e., $\sigma_{\rm E}^2\rightarrow 0$, where the secure communication is guaranteed in the worst case of a powerful Eve. To this end, the ergodic secrecy rate is given by
{\abovedisplayskip=4pt plus 2pt minus 6pt
\abovedisplayshortskip=0pt plus 2pt
\belowdisplayskip=4pt plus 2pt minus 6pt
\belowdisplayshortskip=2pt plus 2pt minus 4pt
\begin{equation}
		R_{\rm sec}=[R_k-C]^+,
\end{equation}}where $[x]^+={\rm max}\{0, x\}$. However, evaluating the expected value in (8) analytically is cumbersome. As an alternative, a lower bound for the ergodic secrecy rate is presented in the following theorem.

\begin{theorem}
	 In the RIS-aided secure system with AN, the ergodic secrecy rate of user $k$ can be evaluated by
	 {\abovedisplayskip=4pt plus 2pt minus 6pt
\abovedisplayshortskip=0pt plus 2pt
\belowdisplayskip=4pt plus 2pt minus 6pt
\belowdisplayshortskip=4pt plus 2pt minus 4pt
	\begin{equation}
		R_{\rm sec}=[\log_2(1+\gamma_k)-\log_2(1+\gamma_{\rm E})]^{+},
	\end{equation}
	with
	\begin{small}
	\begin{align}
		&\gamma_k=S_k/I_k,\ \gamma_{\rm E}=S_{\rm E}/I_{\rm E},\ S_k=\xi P\big[{\rm tr}(\mathbf{R}_k)\big]^2,\\
		&I_k=\xi P\sum\nolimits_{i\neq k}{\rm tr}(\mathbf{R}_k\mathbf{R}_i)+\sigma_k^2\sum\nolimits_{j=1}^{K}{\rm tr}({\bf R}_j),\\
		&S_{\rm E}=\xi MM_{\rm E}(M-K){\rm tr}\big({\bf R}_k({\bf R}_{\rm E}+\beta_3{\bf R}_{\rm B,E})\big),\\
		&I_{\rm E}=(1-\xi)(M-K-M_{\rm E}){\rm tr}({\bf R}_{\rm E}+\beta_3{\bf R}_{\rm B,E})\sum\limits_{j=1}^{K}{\rm tr}({\bf R}_j),
	\end{align}
	\end{small}}where ${\bf R}_k=\beta_{2,k}\mathbf{R}_{{\rm B},k}+\beta_{{\rm I},k}{\bf H}_1{\bf \Phi}\mathbf{R}_{{\rm I},k}{\bf \Phi}^H{\bf H}_1^H$ and $\mathbf{R}_{\rm E}=\beta_{\rm I, E}{\bf H}_1{\bf \Phi}\mathbf{R}_{\rm I, E}{\bf \Phi}^H{\bf H}_1^H$.
\end{theorem}
\begin{proof}
	See Appendix A. 
\end{proof}

\begin{remark}
	It is observed from Theorem 1 that the ergodic secrecy rate depends only on the statistical CSI of the users and Eve, phase shifts $\bf\Phi$, and power fraction $\xi$, motivating further optimization concerning $\bf \Phi$ and $\xi$ at the BS.
\end{remark}

\begin{corollary}
For uncorrelated Rayleigh fading, i.e., $\mathbf{R}_{{\rm I},k}=\mathbf{R}_{\rm I,E}={\bf I}_N$ and $\mathbf{R}_{{\rm B},k}=\mathbf{R}_{\rm B,E}={\bf I}_M$, we obtain $\gamma_k$ and $\gamma_{\rm E}$ as (15) and (16) at the bottom of the next page. From (15), we observe that the inter-user interference always exists even with an infinite number of BS antennas $M$. This is because the cascaded channels through the RIS for the multiple users are not asymptotically orthogonal due to the common component ${\bf H}_1$. In addition, the RIS's ability to modify the wireless medium is significantly impeded since the secrecy rate becomes independent of the RIS phase shifts $\bf \Phi$ but only dependent on the size of RIS.
\end{corollary}

\begin{corollary}
When $N\gg M$, we have ${\bf H}_1{\bf H}_1^H\rightarrow\beta_1N{\bf I}_M$. By substituting $\Vert{\bf H}_1{\bf H}_1^H\Vert_2^2=\beta_1^2N^2M$ into (15) and (16), the secrecy rate is given by (17) at the bottom of this page. We evince that the achievable rate of user $k$ increases logarithmically with the number of BS antennas and the capacity of Eve hardly changes with $M$. This implies that a promising secrecy performance gain is achieved for large $N$.
\end{corollary}

\newcounter{TempEqCnt}
\setcounter{TempEqCnt}{\value{equation}}
\begin{figure*}[hb] 
{\abovedisplayskip=4pt plus 2pt minus 4pt
\abovedisplayshortskip=0pt plus 2pt
\belowdisplayskip=4pt plus 2pt minus 6pt
\belowdisplayshortskip=4pt plus 2pt minus 4pt
	\vspace{-0.5cm}
	\hrulefill
	\begin{small}
	\begin{equation}
	\gamma_k=\frac{\xi P \left(\beta_{{\rm I},k}^2\beta_1^2M^2N^2+\beta_{{\rm I},k}\beta_{2,k}\beta_1M^2N+\beta_{2,k}^2M^2\right)}   {\xi P\sum\nolimits_{i\neq k}\big[\beta_{2,k}\beta_{2,i}M+\left(\beta_{2,k}\beta_{{\rm I},i}+\beta_{{\rm I},k}\beta_{2,i}\right)\beta_1MN+\beta_{{\rm I},k}\beta_{{\rm I},i}\Vert{\bf H}_1{\bf H}_1^H\Vert_2^2\big]+\sigma_k^2\sum\nolimits_{j=1}^{K}\big[\beta_{{\rm I},j}\beta_1MN+\beta_{2,j}M\big]}
	\end{equation}

	\begin{equation}
	\gamma_{\rm E}=\frac{\xi M_{\rm E}(M-K)\left(
	\beta_{{\rm I},k}\beta_{\rm I,E}\Vert{\bf H}_1{\bf H}_1^H\Vert_2^2+\left(\beta_{{\rm I},k}\beta_3+\beta_{2,k}\beta_{\rm I,E}\right)\beta_1MN+
	\beta_{2,k}\beta_3M\right)}  {(1-\xi)(M-K-M_{\rm E})\left(\beta_3+\beta_{\rm I,E}\beta_1N\right)\sum\nolimits_{j=1}^{K}\big[\beta_{{\rm I},j}\beta_1MN+\beta_{2,j}M\big]}
	\end{equation}
	\begin{equation}
	R_{\rm sec}=\left[\log_2\left(1+\frac{M\beta_{{\rm I},k}\beta_1^2}{\sum\nolimits_{i\neq k}\beta_{{\rm I},i}}\right)-
\log_2\left(1+\frac{\xi M_{\rm E}(M-K)
	\beta_{{\rm I},k}}{(1-\xi)(M-K-M_{\rm E})\beta_1^2\sum\nolimits_{j=1}^{K}\beta_{{\rm I},j}}\right)\right]^{+}
	\end{equation}
	\end{small}}

\end{figure*}

\vspace{0.1cm}
\begin{corollary}
Without the existence of RIS, i.e., $\beta_{{\rm I}, k}=0$ and $\beta_{\rm I, E}=0$, the ergodic secrecy rate in (10) reduces to
{\abovedisplayskip=4pt plus 2pt minus 6pt
\abovedisplayshortskip=0pt plus 2pt
\belowdisplayskip=4pt plus 2pt minus 6pt
\belowdisplayshortskip=2pt plus 2pt minus 4pt
\begin{small}
\begin{align}
&R_{\rm sec}=\bigg[\log_2\bigg(1+\frac{\xi \beta_{2,k}^2PM^2/\sum\nolimits_{j=1}^{K}\beta_{2,j}}{\xi P\delta\sum\nolimits_{i\neq k}\beta_{2,i}{\rm tr}(\mathbf{R}_{{\rm B},k}\mathbf{R}_{{\rm B},i})+\sigma_k^2}\bigg)\nonumber\\&-
\log_2\bigg(1+\frac{\xi M_{\rm E}(M-K){\rm tr}\big({\bf R}_{{\rm B},k}{\bf R}_{\rm B,E}\big)}{(1-\xi)M(M-K-M_{\rm E})\sum\nolimits_{j=1}^{K}\beta_{2,j}}\bigg)\bigg]^{+},
\end{align}
\end{small}}where $\delta=\beta_{2,k}/\sum\nolimits_{j=1}^{K}\beta_{2,j}$. Specifically, when the spatial correlation at the BS disappears, the derived $R_{\rm sec}$ in (18) retrieves the result in [14, Theorem 1] as a special case.
\end{corollary}

\vspace{-0.3cm}
\section{Proposed Design for Secrecy Rate Maximization}
In this section, we study the joint optimization of the power fraction $\xi$ and phase shifts $\bf \Phi$ to maximize the ergodic secrecy rate in (10). Mathematically, the optimization problem is formulated as
{\abovedisplayskip=4pt plus 2pt minus 6pt
\abovedisplayshortskip=0pt plus 2pt
\belowdisplayskip=4pt plus 2pt minus 6pt
\belowdisplayshortskip=0pt plus 2pt minus 4pt
\begin{align}
\rm{(P1)}\ \ &\underset{\xi,\bf \Phi}{\rm max} \ \ R_{\rm sec}(\xi,\bf \Phi)  
\\&{\rm s.t.}\ \ \ \ \xi\in[0,1];\ |\phi_i|=1,\ i=1,...,N,\nonumber
\end{align}}where $\phi_i=\exp(j\theta_i)$. It is challenging to jointly optimize $R_{\rm sec}(\xi,\bf \Phi)$ as it is a non-convex function of $\xi$ and $\bf \Phi$. To address this, the alternating optimization (AO) technique is applied to optimize $\xi$ and $\bf \Phi$ by executing refinement processes with efficient closed-form calculations at the BS.
%\textcolor{blue}{To maintain the power constraints, we should have ${\rm tr}({\bf WW}^H)=K$ and ${\rm tr}({\bf V}{\bf V}^H)=M-K$. Note that the design of ${\bf W}$ and ${\bf V}$ in section III can meet the power constraints.}

First, we consider the optimization of $\xi$ by fixing $\bf \Phi$. The following lemma provides a closed-form solution to the fixed-point equation for solving (P1).
\begin{lemma}
	For given $\bf \Phi$, the optimal solution of $\xi$ is
{\abovedisplayskip=4pt plus 2pt minus 6pt
\abovedisplayshortskip=0pt plus 2pt
\belowdisplayskip=4pt plus 2pt minus 6pt
\belowdisplayshortskip=2pt plus 2pt minus 4pt
	\begin{small}
	\begin{equation}
		\xi^*=\frac{-b+\sqrt{b^2-4ac}}{2a},
	\end{equation}
	\end{small}}where $a=B_1(A_1A_2+A_1A_3+A_2^2)-A_1A_3$, $b=2A_3A_1$, and $c=B_1A_3^2-A_1A_3$ are constants with respect to statistical channel spatial correlation matrices.
\end{lemma}
\begin{proof}
	 By taking the first derivative of $R_{\rm sec}$ in (10), it yields
	 {\abovedisplayskip=4pt plus 2pt minus 6pt
\abovedisplayshortskip=0pt plus 2pt
\belowdisplayskip=4pt plus 2pt minus 6pt
\belowdisplayshortskip=2pt plus 2pt minus 4pt
	\begin{small}
	\begin{align}
	R'_{\rm sec}=\frac{\partial R_{\rm sec}}{\partial \xi}=&\frac{A_1A_3}{\ln2(A_3+A_2\xi)[A_3+(A_1+A_2)\xi]}\nonumber\\
	&-\frac{B_1}{\ln2(\xi-1)[1+(B_1-1)\xi]},
	\end{align}
	\end{small}}where $B_1\triangleq \frac{MM_{\rm E}(M-K)\zeta^2{\rm tr}\big({\bf R}_k({\bf R}_{\rm E}+\beta_3{\bf R}_{\rm B,E})\big)}{K(M-K-M_{\rm E}){\rm tr}({\bf R}_{\rm E}+\beta_3{\bf R}_{\rm B,E})\sum\nolimits_{j=1}^{K}{\rm tr}({\bf R}_j)}$, $A_1\triangleq P{\rm tr}(\mathbf{R}_k)^2$, $A_2\triangleq P\sum\nolimits_{i\neq k}{\rm tr}(\mathbf{R}_k\mathbf{R}_i)$, and $A_3\triangleq \sigma_k^2\sum\nolimits_{j=1}^{K}{\rm tr}({\bf R}_j)$. Since $\xi\in[0,1]$, after some algebraic manipulations, it is easily checked that $R''_{\rm sec}<0$, which implies that $R'_{\rm sec}$ is a strictly decreasing function on $\xi$. Moreover, we have $R'_{\rm sec}>0$ for small $\xi$, while $R'_{\rm sec}<0$ for large $\xi$. Hence, there exists an optimal choice of $\xi$ achieving the unique maximum of secrecy rate. Therefore, considering the concavity of $R_{\rm sec}$ with respect to $\xi$, the optimal power fraction in (20) is obtained by solving $R'_{\rm sec}=0$.
\end{proof}

Then, we optimize the RIS phase matrix $\bf \Phi$ for fixed $\xi$, which is less tractable due to the unit-modulus constraints. Due to the complicated form of $R_{\rm sec}$ in (10), we apply the projected gradient ascent method to obtain a locally optimal solution, eventually converging to a stationary point [11]. Specifically at the $l$th step, denote by ${\bf v}^l=[\phi_1^l,...,\phi_n^l,...,\phi_N^l]^T$ the induced phases and by ${\bf q}^k$ the adopted ascent direction, where $[{\bf q}^l]_n=\frac{\partial R_{\rm sec}}{\partial \phi_n^*}$ with respect to $\phi_n=e^{j\theta_n}$ is obtained in the following \emph {Lemma 2}. The subsequent $(l+1)$th iteration step is updated according to
{\abovedisplayskip=4pt plus 2pt minus 6pt
	\abovedisplayshortskip=0pt plus 2pt
	\belowdisplayskip=4pt plus 2pt minus 6pt
	\belowdisplayshortskip=2pt plus 2pt minus 4pt
\begin{equation}
	{\tilde{\bf v}}^{l+1}={\bf v}^{l}+\mu_k{\bf q}^l\ {\text {and}}\ {\bf v}^{l+1}=\exp\big(j\arg\big({\tilde{\bf v}}^{l+1}\big)\big),
\end{equation}}where $\mu_k$ is the step size computed at each step. 

\begin{lemma}
	The gradient of the ergodic secrecy rate, $R_{\rm sec}$, with respect to $\phi_n$ is computed as
	{\abovedisplayskip=4pt plus 2pt minus 6pt
\abovedisplayshortskip=0pt plus 2pt
\belowdisplayskip=4pt plus 2pt minus 6pt
\belowdisplayshortskip=2pt plus 2pt minus 4pt
	\begin{small}
	\begin{equation}
		\frac{\partial R_{\rm sec}}{\partial \phi_n^*}=\frac{1}{\ln2}\bigg(\frac{\frac{\partial \gamma_k}{\partial \phi_n^*}}{1+\gamma_k}-\frac{\frac{\partial \gamma_{\rm E}}{\partial \phi_n^*}}{1+\gamma_{\rm E}}\bigg),
	\end{equation}
	\end{small}}where $\frac{\partial \gamma_k}{\partial \phi_n^*}$ and $\frac{\partial \gamma_{\rm E}}{\partial \phi_n^*}$ are given in (28) and (29), respectively.
\end{lemma}
\begin{proof}
	See Appendix B.
\end{proof}

Now by incorporating \emph{Lemma 1} and the gradient ascent method, concrete steps of the proposed algorithm are summarized in Algorithm~1.

\begin{proposition}
The proposed algorithm always converges to a stationary point of (P1).
\end{proposition}
\begin{proof}
	This is directly checked by the following
	{\abovedisplayskip=4pt plus 2pt minus 6pt
\abovedisplayshortskip=0pt plus 2pt
\belowdisplayskip=4pt plus 2pt minus 6pt
\belowdisplayshortskip=2pt plus 2pt minus 4pt
	\begin{small}
	\begin{equation}
	R_{\rm sec}\big(\xi_{(t)},{\mathbf \Phi}_{(t)}\big)\overset{\rm (a)}{\geq} R_{\rm sec}\big(\xi_{(t-1)},{\mathbf \Phi}_{(t)}\big)\overset{\rm (b)}{\geq} R_{\rm sec}\big(\xi_{(t-1)},{\mathbf \Phi}_{(t-1)}\big),
	\end{equation}
	\end{small}}where $\rm (a)$ holds since the optimization of $\xi$ is convex for given $\bf\Phi$, and $\rm (b)$ holds because the gradient search is along a monotonically increasing direction of $R_{\rm sec}$ [15].
\end{proof}

\begin{algorithm}[t]
	{\small
	\caption{Proposed algorithm for solving P1}
	\label{alg:4}
	\begin{algorithmic}[1]
	\State {\bf Initialize:} 
	$\mathbf{v}^0=\exp(j\pi/2)\mathbf{1}_N$, $\mathbf{\Phi}^0={\rm diag}(\mathbf{v}^0)$, $R_{\rm sec}^0=f(\xi,\mathbf{\Phi}^0)$ given by (9), $\xi\in[0,1]$, $t=0$, and $\epsilon>0$.\\
	{\bf Repeat}\ \ $t\leftarrow t+1$\\
    Find $\xi_{(t)}$ with fixed ${\bf \Phi}_{(t-1)}$ as per (20);
	\For {$l=0,1,2,...,$}
       \State Find $[\mathbf{q}^l]_n=\frac{\partial R_{\rm sec}}{\partial \phi_n^*}$, $n=1,...,N$, as per (23) and $\mu$ by backtrack line search [15];
	\State $\tilde{\mathbf{v}}^{l+1}=\mathbf{v}^{l}+\mu\mathbf{q}^{l}$;\ \ \  $\mathbf{v}^{l+1}=\exp(j\arg(\tilde{\mathbf{v}}^{l+1}))$;
	\State  $\mathbf{\Phi}^{l+1}={\rm diag}(\mathbf{v}^{l+1})$;\ \ \ $R_{\rm sec}^{l+1}=f(\xi_{(t)},\mathbf{\Phi}^{l+1})$;
	\State Until $|R_{\rm sec}(\xi_{(t)},\mathbf{\Phi}^{l+1})-R_{\rm sec}(\xi_{(t)},\mathbf{\Phi}^l)|<\epsilon$;
	\State $\mathbf{\Phi}_{(t)}=\mathbf{\Phi}^{l+1}$;
	\EndFor\\
	{\bf Until} $|R_{\rm sec}(\xi_{(t)},{\mathbf \Phi}_{(t)})-R_{\rm sec}(\xi_{(t-1)},{\mathbf \Phi}_{(t-1)})|<\epsilon$.
	\end{algorithmic}
	}
\end{algorithm}

The algorithm comes with low computational complexity because it consists of simple matrix operations. In particular, the complexity of Algorithm~1, depending on the computations involved in updating the power faction in (20) and the gradient in (23), is $\mathcal O(MN^2+NM^2)$, which is lower compared with that of [6] under practical settings.

\section{Numerical Results}
In this section, numerical simulations are provided to validate the effectiveness of the proposed methods. The distance-dependent large-scale path loss coefficient is $\beta=C_0(\frac{d}{D_0})^{-\zeta}$, where $C_0=-20$ dB is the path loss at the reference distance $D_0=1$ m, $d$ represents the individual link distance, and $\zeta$ denotes the path loss exponent. The pass-loss exponents for the RIS-aided links are set as 2 and 2.2 while the pass-loss exponent for the direct links is set as 3. The distance between the BS and RIS is set to be 20 m, and all the users and Eve are assumed to be located in a circular regime, whose center is 50 m away from the RIS and 60 m away from the BS, and the radius is 3 m. The spatial correlation matrices at the BS are generated according to [14] as $[{\bf R}(\rho)]_{i,j}=\rho^{|i-j|}$, while the spatial correlation matrices at the RIS are given as in [7]. The RIS element spacing is given by $d_{\rm H}=d_{\rm V}=\lambda/4$. Also, the signal-to-noise ratio (SNR) is defined as $10\log10(P/\sigma_k^2)$. Unless otherwise specified, we also set $\sigma_k^2=\sigma_{\rm E}^2=-60$ dBm, $\rm SNR=5$ dB, $\rho=0.4$, $N=256$, $M=128$, and $K=8$.

Fig. 1 illustrates that the derived analytical results and numerical results match well for varying number of Eve's antennas. We observe that a higher number of Eve's antennas degrades the secrecy rates as expected. For comparison, we also depict the results with ZF precoding. It is shown that MRT outperforms ZF at low SNRs while for high SNR values ZF attains a higher secrecy rate since ZF offers interference-free communication to users in the high SNR regime. In the case of imperfect CSI, the estimated channel is modeled as ${\bf z}_k=\sqrt{1-\tau^2}\hat{\bf{z}}_k+\tau{\bf e}_k$ by representing ${\bf h}_k={\bf R}_k^{\frac{1}{2}}{\bf z}_k$, where $\hat{\bf{z}}_k$ is an imperfect observation of ${\bf z}_k$, ${\bf e}_k$ is the Gaussian noise, and $0<\tau<1$ characterizes the CSI imperfection. We observe that the secrecy performance loss is marginal with estimation error $\tau=0.1$ in the tested cases.

Fig. 2 depicts the secrecy rate versus the power fraction for $M_{\rm E}=4$, where the optimal value for $\xi$ in (20) is marked by black stars. It is shown that $\xi^*$ is decreasing in the number of BS antennas $M$, i.e., more power should be allocated to AN. This is because the correlation between ${\bf h}_k$ and ${\bf H}_{\rm E}$ becomes strong with growing $M$ due to the increasing dimension of ${\bf H}_1$, resulting in potentially more information leakage to Eve. On the other hand, $\xi^*$ is increasing in the RIS size $N$, since the effective degree of freedom of the channels from RIS to users increases with $N$. In this case, it can be useful to allocate less power to AN for improving the secrecy performance.
%Interestingly, when $\xi$ is not optimized, increasing $M$ even degrades the secrecy rate. This is because the channels of the users and Eve become highly correlated with large $M$, resulting in more information leakage.

\begin{figure*}
\begin{minipage}[t]{0.22\textwidth}
\centering
\includegraphics[width=1.7in]{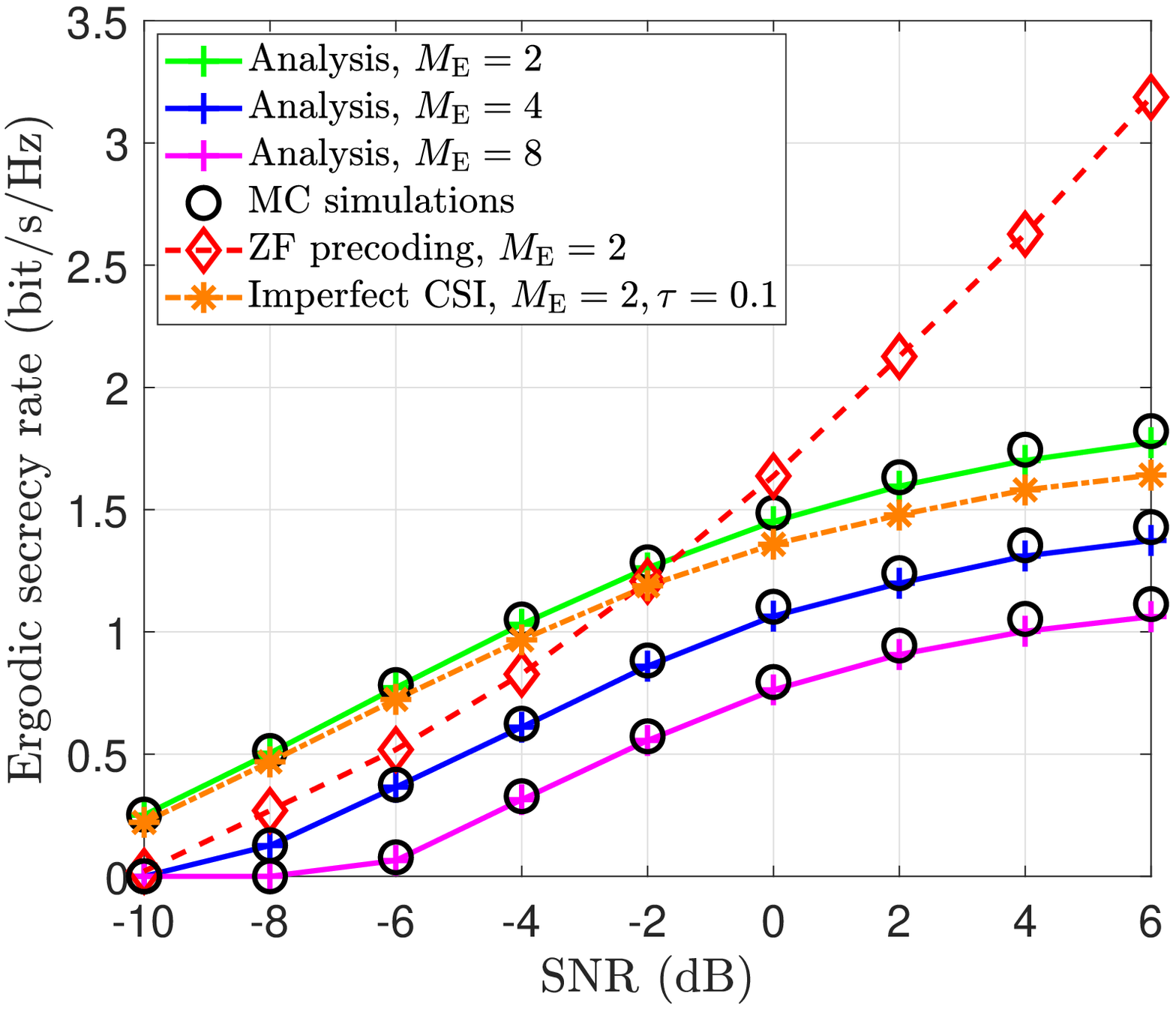}
\caption{Ergodic secrecy rate versus SNR}
\end{minipage}
\hspace{3ex}
\begin{minipage}[t]{0.22\textwidth}
\centering
\includegraphics[width=1.7in]{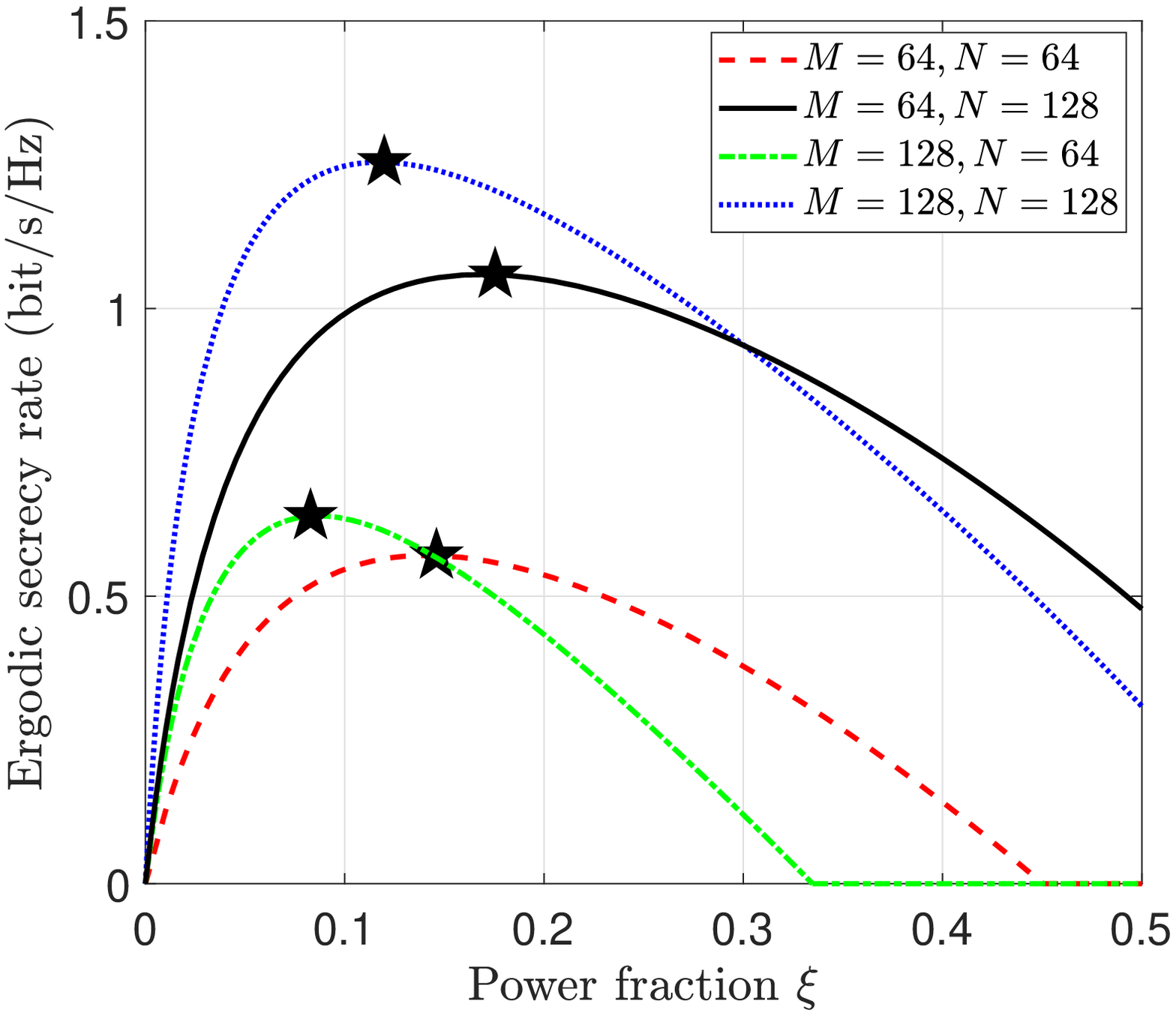}
\caption{Ergodic secrecy rate versus power fraction}
\end{minipage}
\hspace{3ex}
\begin{minipage}[t]{0.22\textwidth}
	\centering
	\includegraphics[width=1.7in]{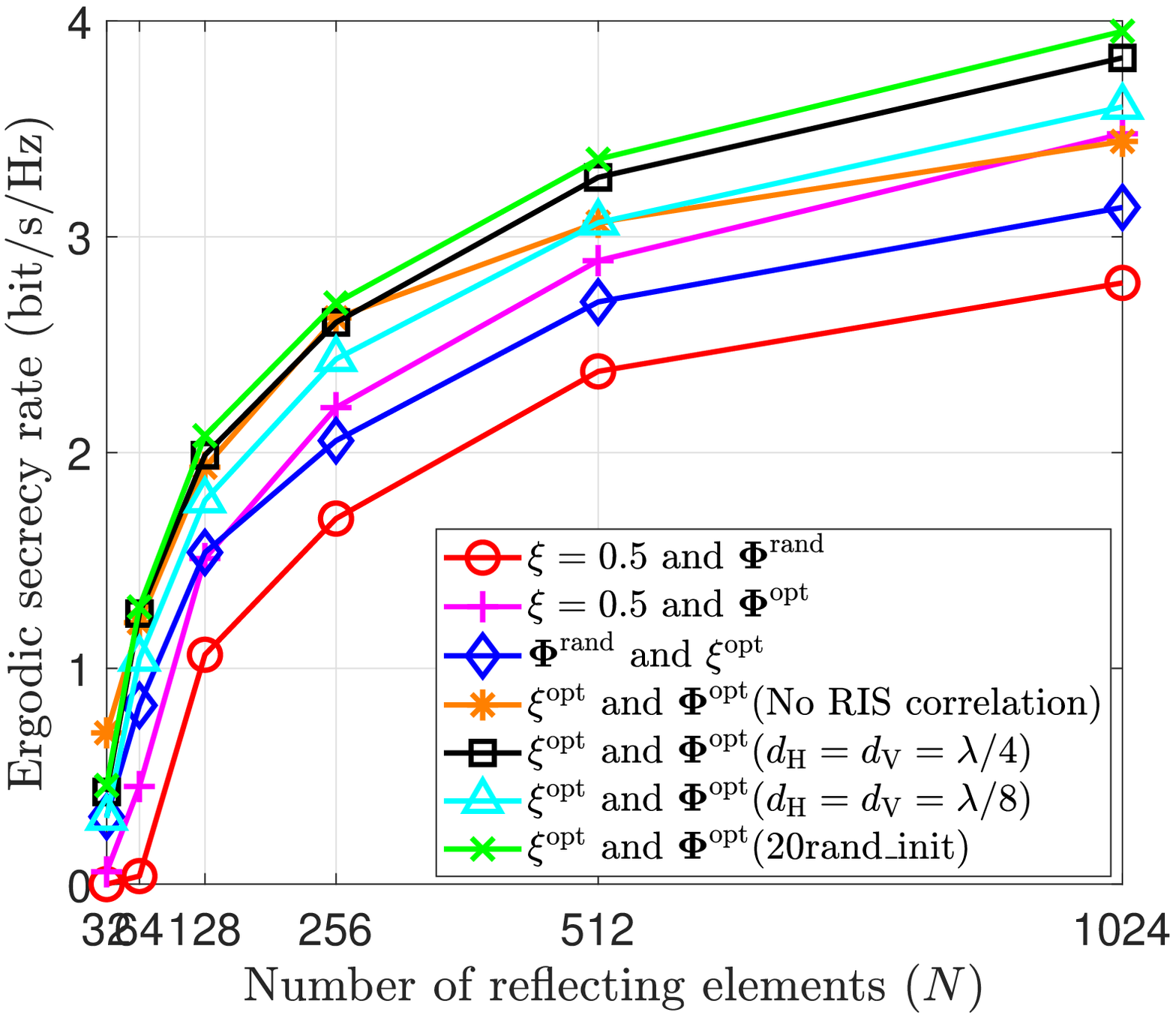}
	\caption{Ergodic secrecy rate versus the number of RIS elements}
\end{minipage}
\hspace{3ex}
\begin{minipage}[t]{0.22\textwidth}
	\centering
	\includegraphics[width=1.7in]{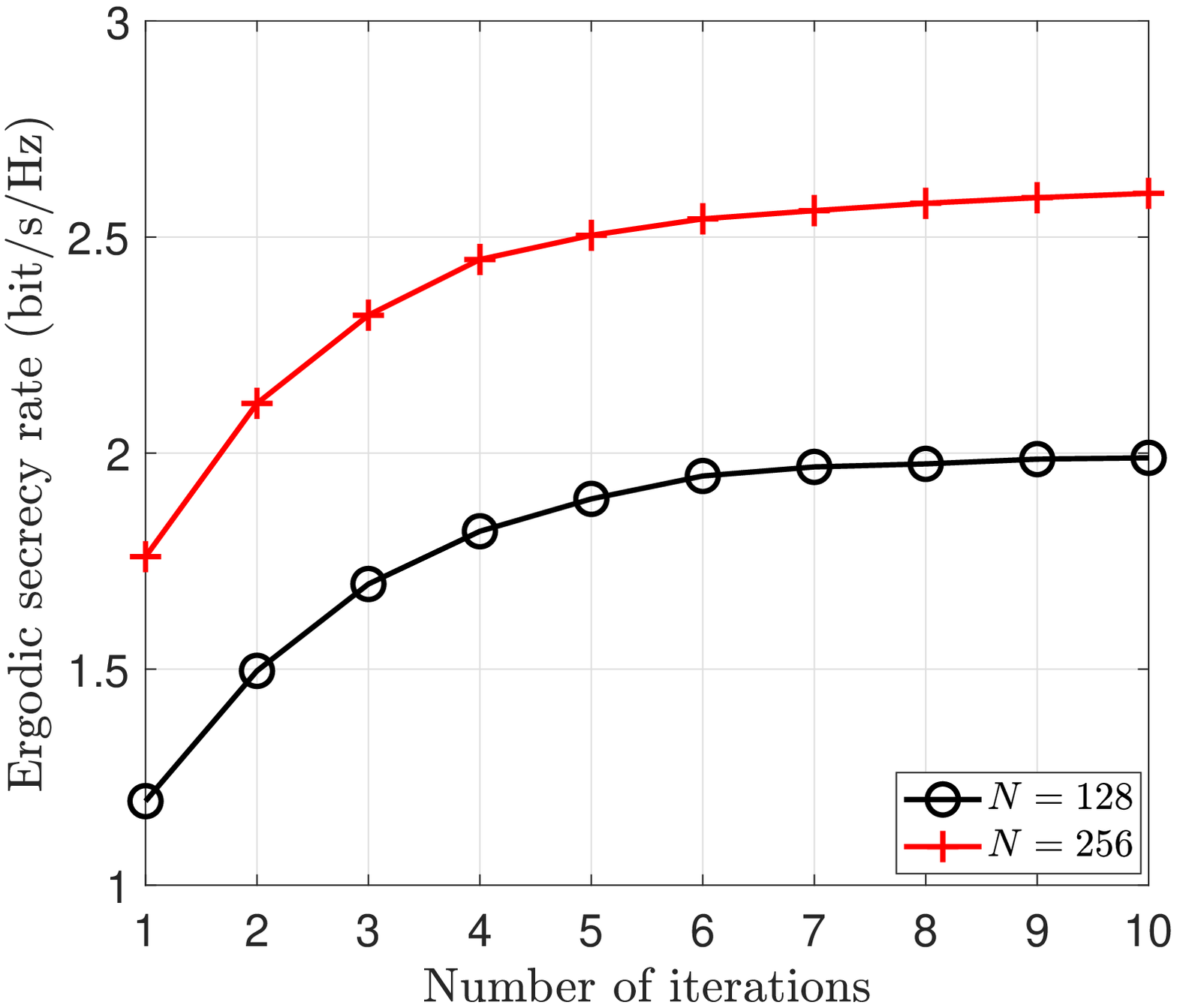}
	\caption{Ergodic secrecy rate versus the number of iterations}
\end{minipage}
\end{figure*}

Fig. 3 presents the secrecy performance of the proposed
scheme versus $N$ for $M_{\rm E}=2$ by using equal power fraction [16] and random phase shifts [7] as benchmarks, i.e., $\xi=0.5$ and $\theta_n\sim\mathcal U[0,2\pi]$. For small $N$, the optimal power fraction with random phase shifts achieves higher secrecy rates than the equal power fraction with random and optimal phase shifts. This is because a small RIS provides limited signal energy boosting for the system where the system tends to be a power-limited scenario and this power fraction optimization plays a dominating role.
In addition, we notice that spatial correlation at the RIS should be taken into account to benefit from the phase shifts design in the case of statistical CSI. Also, it is shown that the secrecy rates decrease when the inter-element spacing reduces from $\lambda/4$ to $\lambda/8$. This is due to an increase in the spatial correlation, which reduces the spatial diversity. We observe a small performance gap by comparing the proposed algorithm with an approximate of the global optimum, which is achieved by running the algorithm twenty times with different initializations and then choosing the best.

Fig. 4 shows the ergodic secrecy rate versus the number of iterations for various numbers of RIS elements with $M_{\rm E}=2$. We observe that the algorithm converges fast in all the tested cases, where the algorithm converges within 10 iterations.

\vspace{-0.3cm}
\section{Conclusion}
We considered the secure communication in RIS-aided multiuser massive MIMO systems. A closed-form expression of the ergodic secrecy rate was derived. Then, based on the expression, we optimized RIS phase shifts and AN power fraction. We showed that a large number of RIS elements and low spatial correlation at the RIS are preferred to achieve high secrecy rates. Future works include extending to Rician and even millimeter-wave channels.

\section*{Appendix A}
For the sake of exposition, we denote the cascade channel of user $k$ by ${\bf h}_k=\sqrt{\beta_{{\rm I},k}}{\bf H}_1{\bf \Phi}\mathbf{R}_{{\rm I},k}^{1/2}{\bf g}_{{\rm I},k}+\sqrt{\beta_{2,k}}\mathbf{R}_{{\rm B},k}^{1/2}{\bf g}_{{\rm B},k}$. Since ${\bf g}_{{\rm I},k}$ and ${\bf g}_{{\rm B},k}$ are independent random vectors, we have that ${\bf h}_k$ follows the complex Gaussian distribution, i.e., ${\bf h}_k\sim\mathcal{CN}(0, {\bf R}_k)$, where ${\bf R}_k=\beta_{2,k}\mathbf{R}_{{\rm B},k}+\beta_{{\rm I},k}{\bf H}_1{\bf \Phi}\mathbf{R}_{{\rm I},k}{\bf \Phi}^H{\bf H}_1^H$.

1) Compute $R_k$: Consider MRT satisfying ${\rm tr}({\bf W}{\bf W}^H)=K$, which leads to ${\bf W}=\sqrt{\frac{K}{{\sum\nolimits_{j=1}^{K}{\rm tr}({\bf R}_j)}}}{\bf H}$.
First, we directly obtain $\left|{\mathbb E}\{{\bf h}_k^H{\bf w}_k\}\right|^2=\frac{K}{{\sum\nolimits_{j=1}^{K}{\rm tr}({\bf R}_j)}}\left[{\rm tr}\left({\bf R}_k\right)\right]^2$ and ${\mathbb E}\left\{\left|{\bf h}_k^H{\bf w}_i\right|^2\right\}=\frac{K}{{\sum\nolimits_{j=1}^{K}{\rm tr}({\bf R}_j)}} {\rm tr}\left({\bf R}_k{\bf R}_i\right)$.
Then, the variance is calculated as
{\abovedisplayskip=4pt plus 2pt minus 6pt
\abovedisplayshortskip=0pt plus 2pt
\belowdisplayskip=4pt plus 2pt minus 6pt
\belowdisplayshortskip=2pt plus 2pt minus 4pt
\begin{small}
\begin{align}
&\frac{1}{M^2}{\rm var}\left\{{\bf h}_k^H{\bf w}_k\right\}\nonumber\\
=&\frac{K}{{\sum\nolimits_{j=1}^{K}{\rm tr}({\bf R}_j)}}\mathbb E\left\{\left|\frac{1}{M}{\bf h}_k^H{\bf h}_k-\frac{1}{M}{\rm tr}({\bf R}_k)\right|^2\right\}\xrightarrow{M\rightarrow\infty} 0,
\end{align}
\end{small}}where (25) is obtained according to [17, Lemma 4]. For the term of ${\mathbb E}\left\{{\bf h}_k^H{\bf V}{\bf V}^H{\bf h}_k\right\}$, it is obviously zero due to the null-space AN method.

2) Compute $C$: To begin with, we rewrite ${\bf X}$ in (8) as ${\bf X}=q{\bf X}_1+q{\bf X}_2$, where ${\bf X}_1\triangleq{\bf H}_{\rm B,E}^H{\bf V}{\bf V}^H{\bf H}_{\rm B,E}$ and ${\bf X}_2\triangleq({\bf H}_1{\bf \Phi}{\bf H}_{\rm I,E})^H{\bf V}{\bf V}^H({\bf H}_1{\bf \Phi}{\bf H}_{\rm I,E})$ are uncorrelated due to the definition in (2). Eigendecompose ${\bf R}_{\rm B, E}={\bf U\Lambda U}^H$ to decorrelate the channel matrix ${\bf H}_{\rm B, E}$ as ${\bf Z}={\bf H}_{\rm B, E}{\bf\Lambda}^{-1/2}{\bf U}^H$, where ${\bf\Lambda}={\rm diag}(\lambda_1,...,\lambda_N)$ contains the eigenvalues of $\bf R$ and the columns of $\bf U$ are the corresponding eigenvectors. Since $\bf U$ is unitary, the statistics of $\bf ZU$ are identical to those of $\bf Z$. Hence, the distribution of ${\bf X}_1$ is the same as
{\abovedisplayskip=4pt plus 4pt minus 6pt
	\abovedisplayshortskip=0pt plus 4pt
	\belowdisplayskip=4pt plus 4pt minus 6pt
	\belowdisplayshortskip=0pt plus 4pt minus 4pt
\begin{small}
\begin{equation}
	\sum\nolimits_{i=1}^{N}\sum\nolimits_{j=1}^{N}\lambda_i^{1/2}\lambda_j^{1/2}{\bf z}_i{\bf v}_i{\bf v}_j^H{\bf z}_j^H,
\end{equation}
\end{small}}where ${\bf z}_i$ is the $i$th row of $\bf Z$ and ${\bf v}_i$ is the $i$th column of $\bf V$. Considering that ${\bf z}_i$ and ${\bf v}_i$ are independent, it is known from [18] that $\sum_{n=1}^{N}\lambda_n{\bf z}_n{\bf v}_n{\bf v}_n^H{\bf z}_n^H$ follows a Wishart distribution, i.e., $\sum_{n=1}^{N}\lambda_n\mathcal{W}_{M_{\rm E}}(M-K,\frac{1}{M}{\bf I}_{M_{\rm E}})$. The distribution of ${\bf X}_2$ is obtained analogously by rewriting ${\bf H}_1{\bf \Phi}{\bf H}_{\rm I,E}=\mathbf{R}_{\rm E}^{1/2}\mathbf{G}_{\rm I,E}$ with $\mathbf{R}_{\rm E}=\beta_{\rm I, E}{\bf H}_1{\bf \Phi}\mathbf{R}_{\rm I, E}{\bf \Phi}^H{\bf H}_1^H$. 
Then, by applying the Jensen's inequality, the capacity of Eve is bounded as\begin{small}
\begin{align}
	C&\le\log_2\big(1+p\mathbb{E}\big\{{\bf w}_k^H{\bf H}_{\rm E}{\bf X}^{-1}{\bf H}_{\rm E}^H{\bf w}_k\big\}\big)\nonumber
	\\&\overset{\rm (a)}{=}\log_2\bigg(1+\frac{\xi M(M-K)\mathbb{E}\big\{{\bf w}_k^H{\bf H}_{\rm E}{\bf H}_{\rm E}^H{\bf w}_k\big\}}{K(1-\xi)(M-K-M_{\rm E}){\rm tr}({\bf R}_{\rm E}+\beta_3{\bf R}_{\rm B,E})}\bigg)\nonumber
	\\&\overset{\rm (b)}{=}\log_2\bigg(1+\frac{\xi \zeta^2MM_{\rm E}(M-K){\rm tr}\big({\bf R}_k({\bf R}_{\rm E}+\beta_3{\bf R}_{\rm B,E})\big)}{K(1-\xi)(M-K-M_{\rm E}){\rm tr}({\bf R}_{\rm E}+\beta_3{\bf R}_{\rm B,E})}\bigg),
\end{align}
\end{small}where $\rm (a)$ uses the property that ${\bf A}^{-1}\xrightarrow{{\rm a.s.}}1/(n-m){\bf I}_m$ for a Wishart matrix ${\bf A}\sim\mathcal{W}_m(n,{\bf I}_m)$ with $n>m$ [18, Sec 2.1.6] and $\sum_{n=1}^{N}\lambda_n={\rm tr}({\bf R}_{\rm E})$, and $\rm (b)$ results form $\mathbb{E}\big\{{\bf w}_k^H{\bf H}_{\rm E}{\bf H}_{\rm E}^H{\bf w}_k\big\}=\zeta^2M_E{\rm tr}\big({\bf R}_k({\bf R}_{\rm E}+\beta_3{\bf R}_{\rm B,E})\big)$ where $\zeta^2=K\big/\big(\sum\nolimits_{j=1}^{K}{\rm tr}({\bf R}_j)\big)$.

\section*{Appendix B}
Using the standard quotient rule of derivatives, we have
\begin{small}
\begin{equation}
	\frac{\partial \gamma_k}{\partial \phi_n^*}=\frac{1}{I_k^2}\left(I_k\frac{\partial S_k}{\partial \phi_n^*}-S_k\frac{\partial I_k}{\partial \phi_n^*}\right),
\end{equation}
\begin{equation}
	\frac{\partial \gamma_{\rm E}}{\partial \phi_n^*}=\frac{1}{I_{\rm E}^2}\left(I_{\rm E}\frac{\partial S_{\rm E}}{\partial \phi_n^*}-S_{\rm E}\frac{\partial I_{\rm E}}{\partial \phi_n^*}\right).
\end{equation}
\end{small}For simplicity, we use the notation $(\cdot)'$ to represent the partial derivative with respect to $\phi_n^*$. Specifically, the term $S_k'$ is given by $S_k'=2\xi P {\rm tr}({\bf R}_k){\rm tr}(\mathbf{R}_k')$, which requires a further derivation of ${\rm tr}(\mathbf{R}_k')$. Since all terms in $\mathbf{R}_k$ depend on $\phi_n^*$, we have
\abovedisplayskip=4pt plus 2pt minus 6pt
\abovedisplayshortskip=0pt plus 2pt
\belowdisplayskip=4pt plus 2pt minus 6pt
\belowdisplayshortskip=2pt plus 2pt minus 4pt
\begin{small}
\begin{align}
&{\rm tr}\big({\bf R}_k'\big)={\rm tr}\left(\beta_{2,k}\frac{\partial \mathbf{R}_{{\rm B},k}}{\partial \phi_n^*}+\beta_{{\rm I},k}\frac{\partial \left({\bf H}_1{\bf \Phi}\mathbf{R}_{{\rm I},k}{\bf \Phi}^H{\bf H}_1^H\right)}{\partial \phi_n^*}
\right)\nonumber\\
\overset{\rm (a)}{=}&\beta_{{\rm I},k}\sum\nolimits_{i,j}\big[{\bf H}_1{\bf \Phi}{\bf R}_{{\rm I},k}\big]_{j,n}[{\bf H}_1^H]^T_{i,n}=\beta_{{\rm I},k}\big[{\bf H}_1^H{\bf H}_1{\bf \Phi}{\bf R}_{{\rm I},k}\big]_{n,n},
\end{align}
\end{small}where $\rm (a)$ is obtained by using Lemma 1 in [8].
To this end, the partial derivatives of $S_k$, $I_k$, $S_{\rm E}$, and $I_{\rm E}$ with respect to $\phi_n^*$ in (28) and (29) are expressed as follows: $S_k'=2\xi P\beta_{{\rm I},k}{\rm tr}({\bf R}_k)\big[{\bf H}_1^H{\bf H}_1{\bf \Phi}{\bf R}_{{\rm I},k}\big]_{n,n}$,
$I_k'=\xi P\sum\nolimits_{i\neq k}\big[{\bf H}_1^H\big(\beta_{{\rm I},k}\mathbf{R}_k{\bf H}_1{\bf \Phi}{\bf R}_{{\rm I},k}+\beta_{{\rm I},i}\mathbf{R}_i{\bf H}_1{\bf \Phi}{\bf R}_{{\rm I},i}\big)\big]_{n,n}+\sigma_k^2\sum\nolimits_{j=1}^{k}\big[{\bf H}_1^H{\bf H}_1{\bf \Phi}{\bf R}_{{\rm I},j}\big]_{n,n}$,
$S_{\rm E}'=\xi MM_{\rm E}(M-K)\big[{\bf H}_1^H\big(\beta_{{\rm I},k}\big({\bf R}_{\rm E}+\beta_3{\bf R}_{\rm B,E}\big){\bf H}_1{\bf \Phi}{\bf R}_{{\rm I},k}+\beta_{{\rm I,E}}\mathbf{R}_k{\bf H}_1{\bf \Phi}{\bf R}_{{\rm I,E}}\big)\big]_{n,n}$, and $I_{\rm E}'
=(1-\xi)(M-K-M_E)\bigg[\beta_{{\rm I,E}}\sum\nolimits_{j=1}^{k}{\rm tr}({\bf R}_j)\big[{\bf H}_1^H{\bf H}_1{\bf \Phi}{\bf R}_{{\rm I,E}}\big]_{n,n}+{\rm tr}({\bf R}_{\rm E}+\beta_3{\bf R}_{\rm B,E})\sum\nolimits_{j=1}^{k}\big[{\bf H}_1^H{\bf H}_1{\bf \Phi}{\bf R}_{{\rm I},j}\big]_{n,n}\bigg]$.

\end{document}